\documentclass{article}

\usepackage{amsmath}

\usepackage{times}
\usepackage{bm}
\usepackage{natbib}

\usepackage[plain,noend]{algorithm2e}

\makeatletter
\renewcommand{\algocf@captiontext}[2]{#1\algocf@typo. \AlCapFnt{}#2} 
\def\@algocf@capt@plain{top}
\renewcommand{\algocf@makecaption}[2]{%
  \addtolength{\hsize}{\algomargin}%
  \sbox\@tempboxa{\algocf@captiontext{#1}{#2}}%
  \ifdim\wd\@tempboxa >\hsize
    \hskip .5\algomargin%
    \parbox[t]{\hsize}{\algocf@captiontext{#1}{#2}}
  \else%
    \global\@minipagefalse%
    \hbox to\hsize{\box\@tempboxa}
  \fi%
  \addtolength{\hsize}{-\algomargin}%
}
\makeatother

\def\Bka{{\it Biometrika}}

\usepackage{graphicx}
\usepackage{url} 
\usepackage{verbatim}
\usepackage{multirow}
\usepackage[left]{showlabels}
\usepackage{color, xcolor}
\usepackage{longtable}
\usepackage{enumerate}

\DeclareMathOperator{\graph}{\mathcal{A}}
\DeclareMathOperator{\node}{\mathcal{V}}
\DeclareMathOperator{\edge}{\mathcal{E}}

\newtheorem{theorem}{Theorem}

\newtheorem{definition}{Definition}
\newtheorem{algo}{Algorithm}

\newtheorem{remark}{Remark}
\newtheorem{lemma}[theorem]{Lemma}

\newcommand{\q}{1 - O(1/n)}
\newcommand{\calM}{\mathcal{M}}
\newcommand{\calS}{\mathcal{S}}
\newcommand{\calD}{\mathcal{D}}
\newcommand{\cG}{{\cal G}}

\newcommand{\beq}{\begin{equation}}
\newcommand{\eeq}{\end{equation}}

\newcommand{\Rhat}{\hat{R}}
\newcommand{\bitem}{\begin{itemize}}
\newcommand{\eitem}{\end{itemize}}

\newcommand{\goto}{\rightarrow}

\newcommand{\beqn}{\begin{equation}}
\newcommand{\eeqn}{\end{equation}}
\newcommand{\balign}{\begin{align}}
\newcommand{\ealign}{\end{align}}

\newcommand{\lhat}{\hat{\ell}}

\newcommand{\E}{E}
\newcommand{\Etilde}{\tilde{\E}}

\def\limsup{\mathop{\overline{\rm lim}}}

\newcommand{\be}{\begin{equation}}
\newcommand{\ee}{\end{equation}}
\newcommand{\ba}{\begin{array}}
\newcommand{\ea}{\end{array}}
\newcommand{\OX}{\Omega}
\newcommand{\OXtilde}{\tilde{\Omega}}

\newcommand{\Xihat}{\hat{\Xi}}

\newcommand{\RR}{\mathbb{R}}

\newcommand{\munorm}{R}

\def\ben{\begin{equation*}}
\def\een{\end{equation*}}
\def\bea{\begin{eqnarray}}
\def\eea{\end{eqnarray}}
\newcommand{\name}{spectral clustering on network-adjusted covariates}
\newcommand{\names}{spectral clustering on network-adjusted covariates }

\usepackage{amssymb}

\newcommand{\blue}{\textcolor{black}}

\begin{document}



\markboth{Y. Hu and W. Wang}{Community detection with network-adjusted covariates}

\title{Network-Adjusted Covariates for Community Detection}
\author{Yaofang Hu and Wanjie Wang}
\date{\today}

\maketitle

\begin{abstract}
Community detection is a crucial task in network analysis that can be significantly improved by incorporating subject-level information, i.e. covariates. 
Existing methods have shown the effectiveness of using covariates on the low-degree nodes, but rarely discuss the case where communities have significantly different density levels, i.e. multiscale networks. 
In this paper, we introduce a novel method that addresses this challenge by constructing network-adjusted covariates, which leverage the network connections and covariates with a node-specific weight to each node. 
This weight can be calculated without tuning parameters.
We present novel theoretical results on the strong consistency of our method under degree-corrected stochastic blockmodels with covariates, even in the presence of mis-specification and multiple sparse communities. Additionally, we establish a general lower bound for the community detection problem when both network and covariates are present, and it shows our method is optimal for connection intensity up to a constant factor. Our method outperforms existing approaches in simulations and a LastFM app user network. We then compare our method with others on a statistics publication citation network where $30\%$ of nodes are isolated, and our method produces reasonable and balanced results.
\end{abstract}


\section{Introduction}
\label{sec:intro}
Network data refers to the records of connections or relationships between subjects. It can be found in a large variety of scientific fields \citep{casc, chen2006detecting, sporns2016modular, deco2011dynamical, jacob2011role, gil1996political, leskovec2012learning, ying2018graph}. 
A network data is often represented as a graph $\graph = (\node, \edge)$, where $\node$ denotes the set of $n$ nodes, i.e. subjects, and $\edge$ denotes the set of edges or links between nodes.
Mathematically, $\graph$ with $|\node| = n$ can also be expressed by an adjacency matrix $A \in \{0,1\}^{n\times n}$ , where $A_{ij} = A_{ji} = 1$ if $(i,j) \in \edge$ and $A_{ij} = A_{ji} = 0$ otherwise. 
Studies on the network data provide valuable insights into the structure by exploring interactions among subjects.

Among topics on network analysis, the most influential one is the community detection problem. It is also recognized as the clustering of nodes in the area of network analysis. Communities refer to the groups of nodes, so that nodes within the same community are more densely connected than nodes in different communities. If we consider networks on genome data or brain images, the community structure can represent functional modules or coordination between nodes. 
Therefore, detecting the communities can provide insights into challenging biological problems. 

There are plenty of studies of the algorithms and theoretical limits on the community detection problem, especially for sparse networks. Let $d_i$ denote the number of neighbors of node $i$. It has been found in multiple works that $\min_{i\in \node}E(d_i) \geq C\log n$ for a constant $C > 0$ is required to assure the exact recovery of community labels for each node \citep{DCBickel, SBMlower}. Such limits are often referred to as information theoretical lower bounds. 
How to handle networks with some nodes' expected degree being bounded remains a challenging problem. \cite{joseph2016impact} suggests classifying all sparsely connected nodes as a single community, but this might be an over-simplification in some scenarios. In \cite{lei2020consistency}, multiscale networks where the communities have different density levels have been discussed, but $E(d_i) \geq C\log n$ needs to hold for all levels.

This work studies the possibility of sparse network community detection by leveraging covariate information. We consider a challenging mixture setting where the communities can be either relatively dense or extremely sparse. 
According to the fundamental limits in community detection \citep{DCBickel, SBMlower}, we require the relatively dense communities to have expected degrees larger than $c_d \log n$, where $d_i \gg \log n$ is also allowed. 
Meanwhile, nodes in extremely sparse communities have expected degrees no larger than $c_s\log n$, where $d_i = O(1)$ is allowed.
Mathematically, a multiscale network with both relatively dense and extremely sparse communities is defined as follows. 
\begin{definition}\label{def:comm}
Consider a network $\graph = (\node, \edge)$ and constants $c_d > c_s > 0$.  Consider community $k$, we call it a
\begin{enumerate}
\item {\it (Relatively) dense community}, if $E(d_i) \geq c_d \log n$ for all $i$ so that $\ell(i) = k$; 
\item {\it (Extremely) sparse community}, if $E(d_i) \leq c_s\log n$, for all $i$ so that $\ell(i) = k$.
\end{enumerate}
The network $\graph$ is called a multiscale network with extremely sparse communities if both kinds of communities exist.
\end{definition}

Consider the nodes in sparse communities where the labels cannot be recovered by the network.
Modern datasets often include subject-level covariates other than the network. Let $x_i \in \mathbb{R}^{p}$ denote the covariate vector of node $i$. The covariate matrix is defined as $X = (x_1, \cdots, x_n)' \in \mathbb{R}^{n \times p}$. In biology, these covariates may include demographic information, clinical or genetic data, or other relevant features. The covariates often depend on the community structure of the subjects. Therefore, integrating $X$ with $\graph$ will largely improve the community detection results, especially for nodes in sparse communities.

The integration of covariates with network data for community detection has become popular very recently. In \citet{BayesianNewman}, a model was proposed that incorporates low-dimensional discrete covariates and the network based on community memberships. The maximum likelihood estimate is obtained using the belief propagation package. 
\cite{attrisparse} interprets it as an optimization problem on the weighted sum of network Laplacian and covariates kernel matrix. 
In \cite{casc}, the spectral clustering method was discussed, with an application on a weighted sum of the network Laplacian and covariates. 
Later, the issue of mis-matching between covariates and community labels became a concern.
To solve this problem, \cite{cesna} uses the maximum likelihood approach where covariates and network are considered separately, and \cite{jcdc} optimizes a joint community detection criterion analogous to the modularity. 
Other approaches can be found in various studies, including \cite{FengEdge}, \cite{FengNodal},  \cite{InferenceCov} and \cite{xu2022covariate}.

The theoretical consistency of network community detection algorithms has also been a hot topic in recent years.
For networks, studies first discuss weak consistency, i.e., the clustering error rate converges to 0 as $n \to \infty$. Later strong consistency (i.e., exact recovery) is brought into concern, which means that the label of every node can be exactly recovered with a high probability. Without covariates, weak consistency can be achieved when $E(d_i) \to \infty$ and strong consistency requires $E(d_i) \geq C\log n$ \citep{abbe2015exact,gao2017achieving,abbe2020entrywise, chen2021spectral}. With covariates, weak consistency has been found for algorithms under regular conditions. Further, for the high-dimensional covariates, \cite{contextualsbm} sets up the fundamental limit of signal-to-noise ratio to guarantee weak consistency, and \cite{ma2023community} generalizes the results to multi-layer networks. 
Yet there are very few works on the strong consistency results when covariates are involved. The only work is \cite{fanPCA}, where the two-community stochastic blockmodel is considered. The upper bound of their proposed spectral method and the lower bound for strong consistency have been established.

A direct generalization of
these methods to multiscale networks with covariates faces several challenges. 
Firstly, in multiscale networks, the usefulness of covariates depends on the density of the communities. Indeed, dense communities can be recovered with connections alone, while a successful recovery of sparse communities would depend more on the covariates.
However, most existing algorithms leverage the network and covariates by putting a single weight on the whole covariate matrix $X$, without calibrations on the node-specific effects.
Secondly, to elaborate on the performance of each node in relatively dense communities and sparse communities,
we want to establish the strong consistency results of our algorithm. 
However, a multiscale network will induce multi-scale errors, which is a challenge in theoretical analysis.

We introduce a novel approach called \name, which is tuning-free and efficient on multiscale networks with covariates. This approach contains two steps. We first define the network-adjusted covariate vectors 
\[
y_i = \alpha_i x_i + \sum\nolimits_{j:A_{ij} = 1} x_j, \qquad i \in [n].
\]
The new covariate $y_i$ combines the original covariate $x_i$ and the network information using $\sum_{i:A_{ij} = 1} x_j$. We design the node-specific coefficient $\alpha_i$, so that it effectively balances the contribution of $x_i$ and the neighbors. 
Then we apply spectral clustering on the network-adjusted covariate matrix $Y=[y_1,\ldots,y_N]$.  
Under the degree-corrected stochastic blockmodel \citep{DCSBM, DCzhuji, DCBickel}, we prove novel results on spectral properties of $Y$, where we control the row-wise distance between the population and empirical spectral matrices. It hence induces strong consistency of our new approach. We further set up the lower bound, which meets the upper bound induced by our algorithm up to a constant. 

Our work also considers the challenging scenarios where the covariates can be misspecified.
For node $i$, the covariate $x_i$ may not follow the common covariate distribution of nodes in this community. This mis-specification may come from random error or a systematic mismatching between covariates and community labels. 
By our new spectral analysis results, we find that the node label can be exactly recovered, when the node is either in the relatively dense communities or its covariate is correctly specified. In other words, even with the existence of mis-specification, our method still recovers the node labels as long as the information (either from $\graph$ or $X$) is sufficient. 
Such analysis on each node is novel.


The spectral information is commonly used in various statistical fields, including the community detection; see \cite{npca}, \cite{SC3}, \cite{opca}, \cite{SC1}, and \cite{SCORE}. 
Weak consistency can be proved by controlling the Frobenius norm with the Davis-Kahan theorem \citep{SCORE, lei2015consistency}. Recently, \cite{fan2018eigenvector} has established the upper bound on the $\ell_{\infty}$ norm of the eigenvector perturbation. This improvement has motivated the strong consistency results of spectral methods; see \cite{abbe2020entrywise} and \cite{su2019strong}. 
When the covariates are included, \cite{casc} applies the spectral method with weak consistency and \cite{fanPCA} considers an aggregate spectral method with the strong consistency results. 
Therefore, we consider the spectral clustering for the multiscale networks with covariates. We demonstrate that our method can achieve the exact recovery results, except for low-degree nodes that are mis-specified.

To conclude, this work discusses multiscale networks with covariates that mis-specification is considered. We propose the new network-adjusted covariate vectors that assign node-specific weights to covariates. 
Using these network-adjusted covariates, we propose a tuning-free and computationally efficient community detection method. 
We provide solid theoretical results for the new method. The entry-wise perturbation of eigenvectors shows that the label recovery is based on $\graph$ for nodes in dense communities and $X$ for nodes in sparse communities. Hence, an exact recovery can be achieved when the mis-specification happens only in dense communities. 
We further establish the lower bound for the community detection problem on networks with covariates. 
The lower bound matches the upper bound from our new method up to a constant factor, which suggests the optimality of our approach.

\section{Methodology}
\subsection{Notations and background}\label{sec:notation}
We represent a network $\graph$ with covariates as a duplex $(A,X)$, where $A\in\RR^{n\times n}$ is the adjacency matrix and $X\in \RR^{n\times p}$ is the covariate matrix. Each row in $X$ is $x_i$, the covariate vector associated with node $i$. 
For the adjacency matrix $A$, let $d_i = \sum_{j=1}^n A_{ij}$ denote the degree of node $i$ and $\bar{d} = \sum_{i=1}^n d_i/n$ be the average degree. Let $\mathcal{D}$ denote nodes in relatively dense communities and $\mathcal{S}$ denote nodes in sparse communities.  $\calD$ and $\calS$ are unknown to us. 

Let $K$ be the number of communities and $\ell \in \RR^n$ be the community label vector, where each entry $\ell(i) \in [K]$ corresponds to the community membership of node $i$. It can also be represented in the matrix form $\Pi \in \{0,1\}^{n \times K}$, where $\Pi(i, j) = 1$ if $\ell(i) = j$ and 0 otherwise. 
Our objective is to recover $\ell$.

For a vector $a$, $\|a\|$ gives the $\ell_2$ norm of $a$ and $\|a\|_{\infty} = \max_i\|a_i\|$ gives the $\ell_{\infty}$ norm. 
Let $A$ be a matrix, $\lambda_k(A)$ denotes the $k$-th largest singular value of $A$, and $\|A\| = \lambda_1(A)$. 
For two series $a_n$ and $b_n$, we say $a_n \asymp b_n$ if there is a constant $C$, such that $a_n \leq C b_n$ and $b_n \leq C a_n$ when $n$ is large enough. We say $a_n \lesssim b_n$ if $\limsup_{n \goto \infty} a_n/b_n \leq 1$. We have $a_n \gtrsim b_n$ in a similar way. Finally, we use the notation $[N]:=\{1,\ldots,N\}$ for any integer $N$.

\subsection{Network-adjusted covariates}\label{sec:nac}
To leverage the network and covariates, we propose the network-adjusted covariate vectors:
\be
y_i = \alpha_i x_i + \sum\nolimits_{j:A_{ij} = 1} x_j, \quad i \in [n].
\ee
Here, the weight $\alpha_i$ of the node's covariate is defined as 
\be\label{eqn:alpha}
\alpha_i = \frac{\bar{d}/2}{d_i/{\log n} + 1}.
\ee

The network-adjusted covariate vector $y_i$ consists of two parts: the covariate vector of the node itself, $\alpha_i x_i$, and sum of covariates of its neighbors, $\sum_{j:A_{ij} = 1} x_j$. 
The former conveys the node's individual covariate information, while the latter establishes the node's network information. Similar methods utilizing neighbors can be found in \cite{hu2022graph}. 

Here is the intuition of how $\alpha_i$ balances these two parts.
Consider a multiscale network. When the community sizes are comparable, the average degree $\bar{d}$ is dominated by the dense communities. Let $d_{\calD} \geq c_d\log n$ and $d_{\calS} \leq c_s \log n$ denote the average degree of the dense communities and sparse communities, respectively. Hence $\bar{d} \asymp d_{\calD}$. 
Suppose nodes in the same community have expected degrees at the same asymptotic rate, then $\alpha_i \approx (\log n)/2$ for $i \in \calD$ and $\alpha_i \approx \bar{d} \asymp d_{\calD}$ for $i \in \calS$ as the weight of $x_i$. 
Meanwhile, recall that the number of neighbors is $d_i \asymp d_{\calD}$ for $i \in \calD$ and $d_i \asymp d_{\calS} \leq d_{\calD}$ for $i \in \calS$. 
Therefore, $\|y_i\| \asymp d_{\calD}\|x_i\|$ for all nodes. Further, when $c_{\calD}$ is sufficiently large, $y_i$ focuses on $\sum_{j:A_{ij} = 1} x_j$ when $i \in \calD$ and $\alpha_i x_i$ when $i \in \calS$. Hence, this new network-adjusted covariate vector $y_i$ can efficiently leverage the information from both the network and original covariates.

The definition of $\alpha_i$ in \eqref{eqn:alpha} includes two manually decided factors. The numerator $\bar{d}/2$ has a constant factor of $1/2$. 
We want to point out that this constant factor is not essential and any constant between 0 and 1 can be used instead. The strong consistency results in Section \ref{sec:consistency} hold for any such constants; proofs are in supplementary materials. 
The denominator $d_i/\log n + 1$ has a $\log n$ term. This term is decided by the definition of relatively dense communities and extremely sparse communities, which traces back to the fundamental limits of strong consistency. 
In our numerical tests, we have found that using $1/2$ and $\log n$ yields good results.

Let $Y = (y_1, \cdots y_n)'$ be the network-adjusted covariate matrix, then 
\begin{equation}\label{eqn:anc}
Y = AX + D_{\alpha} X, 
\end{equation}
where $D_{\alpha}$ is a diagonal matrix with diagonals as $(\alpha_1, \cdots, \alpha_n)$. 

Consider the special case that $\graph$ has a uniform scale and $d_i \geq c_d \log n$, i.e., $\node = \calD$. By the formula, we have $D_{\alpha} \approx (\log n/2) I$ and $Y \approx \{A + (\log n/2) I\}X$. The left singular vectors of $Y$ mainly depend on $A$. Hence, community detection based on the left singular vectors yields the same error rate based on $A$.


\subsection{Spectral clustering on network-adjusted covariates}\label{sec:cascore}
Spectral methods on community detection were first proposed in \cite{npca}, and have since been developed in various directions, such as spectral methods on the graph Laplacian, regularized Laplacian, non-backtracking matrix, and more \citep{opca, krzakala2013spectral, joseph2016impact, SC3}. Theoretical discussions about eliminating the degree effects using spectral methods have been shown under degree-corrected stochastic blockmodels in \cite{SCORE} and \cite{lei2015consistency}.
In \cite{casc}, spectral methods are employed for networks with covariates, on a weighted summation of the network Laplacian and covariate matrix, where the weight is a tuning parameter. \cite{fanPCA} has discussed the aggregated spectral methods on the two-community stochastic blockmodel.
Here, we apply spectral clustering on the network-adjusted covariate matrix $Y$ in \eqref{eqn:anc} and propose Algorithm \ref{alg1}.

\begin{algo}\label{alg1}
Spectral clustering on network-adjusted covariates

Input: adjacency matrix $A$, covariate matrix $X$, number of communities $K$.
\begin{enumerate}
    \item Find $Y = AX + D_{\alpha} X$, where $D_{\alpha}$ is defined in \eqref{eqn:alpha}. 
\item Find the top $K$ left singular vectors $\Xihat = [\hat{\xi}_1, \ldots, \hat{\xi}_K]$ of $Y$.
\item Find $\Rhat$ by normalizing $\Xihat$ such that each row has norm 1.
\item Perform $k$-means clustering on $\Rhat$ with $K$ clusters, treating every row as a data point. 
\item The output label vector $\hat{\ell}$ by $k$-means in Step 4 gives us the community label. 
\end{enumerate}
\end{algo}

In Step 4, we apply the built-in {\it kmeans} function in {\it R}, which finds a local optimum by the algorithm in \citet{kmeans}. To reduce errors, we use multiple random seeds.  

Here is the high-level intuition why applying $k$-means on $\Rhat$ yields a satisfactory community detection result. 
As we explained, the network-adjusted covariate vectors $y_i$ are dominated by the covariates of neighbors if $i \in \mathcal{D}$ or by the covariate itself if $i \in \mathcal{S}$. 
Therefore, consider a relatively dense community $k_d$ and all nodes in community $k_d$ share the same distribution of neighbors.
Thus, the sum of neighbors' covariates has the same distribution, up to the degree heterogeneity factor. Therefore,  rows in $Y$ corresponding to community $k_d$ share the same centre up to a constant factor. Now we turn to a sparse community, say $k_s$. For $i$ in community $k_s$, $y_i$ is dominated by $\alpha_i x_i$. Since $x_i$s have the same distribution, up to the constant factor $\alpha_i$, these $y_i$s have the same distribution. Again, rows of $Y$ corresponding to community $k_s$ share the same centre up to a constant factor. 
With a delicate random matrix analysis, we can prove the left singular matrix $\Xihat$ in Step 2 inherits such consistency within each community, and hence rows of $\Rhat$ in Step 3 corresponding to the same community will have the same centre, eliminating the constant factor by normalization.
The $k$-means algorithm minimizes the within-cluster sum of the squared distance to the centre, which achieves the true labels. 


\subsection{Generalization with uninformative covariates}
The newly proposed network-adjusted covariate matrix $Y$ can be seen as a product of $A+D_{\alpha}$ and $X$. By linear algebra, a meaningful $\hat{R}$ requires $X$ to hold some information on the community structure. In other words, $X$ cannot be uninformative; otherwise, involving $X$ is pointless. In most cases, researchers can tell whether this is the case based on their experience. 

But for the sake of completeness, we still take this case into consideration. When it is difficult to decide whether $X$ should be involved, we propose a slightly modified version, Algorithm \ref{alg2} as follows. 
In Algorithm \ref{alg2}, we combine the new covariates matrix $YY'$ and the network $AA'$ by a weighted summation, and then apply spectral clustering on this combined matrix. 
The new term $AA'$ does not rely on $X$. Adding it helps us to handle the extreme scenario when $X$ is uninformative.

\begin{algo}\label{alg2}
Spectral clustering on generalized network-adjusted covariates

Input: adjacency matrix $A$, covariate matrix $X$, number of communities $K$.
\begin{enumerate}
\item Find $Y = AX + D_{\alpha} X$, where $D_{\alpha}$ is defined in \eqref{eqn:alpha}. 
\item Define $L = YY' + \beta n AA'$. 
\item Find the top $K$ left singular vectors $\Xihat = [\hat{\xi}_1, \ldots, \hat{\xi}_K]$ of {$L$}.
\item Find $\Rhat$ by normalizing $\Xihat$ such that each row has norm 1.
\item Perform $k$-means clustering on $\Rhat$ with $K$ clusters, treating every row as a data point. 
\item The output label vector $\hat{\ell}$ by $k$-means in Step 4 gives us the community label. 
\end{enumerate}
\end{algo}
In Section \ref{sec:ainclude}, we show that the oracle matrix of $L$ can be written in the same format of $YY'$, but with a generalized definition of $X$. Therefore, we call it as ``generalized network-adjusted covariates".


\blue{The tuning parameter $\beta$ intends to balance the term $nAA'$ and $YY'$. 
Theoretical analysis in Section \ref{sec:ainclude} suggests $\beta$ to be a constant related to $\|x_i\|^2$. 
For numerical analysis, we choose $\beta = \|\bar{x}\|^2$, where $\bar{x}$ is the average covariate vector. It has shown promising clustering results in data analysis. 
To understand this selection, consider the simplified case that $X$ is uninformative, i.e., all nodes have the same mean covariate vector $\mu = E(x_i)$. For this case, $A$ must be dense and $Y \approx AX$. When $\mu$ overrides the noise in $x_i$,  
it follows that $\|YY'\| \approx \|AXX'A'\| \approx \|A 1_n\mu' \mu 1'_n A'\| \leq  n \|\mu\|^2 \|AA'\|$, where $1_n \in \mathcal{R}^n$ has all entries as 1. 
When $\beta \geq \|\mu\|^2$,  $\|\beta n AA'\| \geq n \|\mu\|^2 \|AA'\| \geq \|YY'\|$. So $L \approx \beta n AA'$, which provides community information. 
This motivates us to use $\beta = \|\mu\|^2$, which becomes $\|\bar{x}\|^2$ as the data version.
For special case that $\|\mu\|$ is much smaller than $\|x_i\|$, which can occur when $\mu = 0$ due to signal cancellation, quantiles of $\{\|x_i\|^2\}_{i\in[n]}$ may be a good choice for $\beta$. Further theoretical discussion of $\beta$ can be found  in Section \ref{sec:ainclude}.}

\section{Theoretical guarantee}\label{sec:theory}
\subsection{Degree corrected stochastic blockmodel with covariates}
To formulate the consistency of our proposed approach,
we first model the network with covariates $(A, X)$, under the assumption that $A$ and $X$ are independent given $\ell$. Then we introduce the relatively dense and sparse communities and mis-specification into the model. 

One of the most popular network models is the (degree-corrected) stochastic blockmodel. The original stochastic blockmodel was first proposed in the seminal work \citep{SBM} and it produced promising community detection results. Later works \citep{DCSBM, DCBickel, DCselection, DCzhuji} generalized it to the degree-corrected stochastic blockmodel, which allows for degree heterogeneity.

We follow the same line to model the network. Say $\graph$ has $n$ nodes in $K$ communities, and the community membership matrix is $\Pi$. 
Define a symmetric matrix $P \in \mathbb{R}^{K \times K}$, where $P(k, l)$ denotes the connection intensity parameter between a node in community $k$ and a node in community $l$. To account for degree heterogeneity, we introduce a diagonal matrix $\Theta$ with diagonals $\Theta_{ii} = \theta_i$, representing the popularity of node $i$. 
The adjacency $A$ has Bernoulli distributed entries with parameter $P(A_{ij} = 1) = \theta_i \theta_j P\{\ell(i), \ell(j)\}$, so the probability of an edge between nodes $i$ and $j$ depends on their popularity and the connection intensity between the communities to which they belong.
In matrix form, the adjacency matrix $A$ can be fully identified by $E(A|\Pi)$, that is 
\be\label{eqn:dcsbm}
E(A|\Pi) = \Omega_A - diag(\Omega_A), \quad \Omega_A = \Theta \Pi P \Pi' \Theta.
\ee
Here, $diag(\Omega_A)$ is the diagonal matrix formed by replacing all the off-diagonals of $\Omega_A$ to be 0. This eliminates the possibility of self-loops in the network. 

Now we model the covariates. Given the label $\ell$, we assume $X$ is independent of $A$. The covariates $x_i$ is generated by a standard cluster model (\citet{jin2016influential, jin2017phase}), that they are independently distributed as 
\be\label{eqn:covcluster}
x_i|\Pi \sim F_{k},\quad \ell(i) = k. 
\ee
Here, $F_k$ is a general distribution for community $k$, $k \in [K]$. 
We further model the mis-specification issue. Let $\calM$ denote the set of mis-specified nodes, which means $x_i$ does not follow $F_k$ for $i \in \calM$ and $\ell(i) = k$. We allow $x_i$ to follow any distribution $G_i$, which can be either $F_k$ that $k \neq \ell(i)$ or other distributions.

Combining \eqref{eqn:dcsbm} and \eqref{eqn:covcluster} gives us the degree-corrected stochastic blockmodel with covariates.  
\begin{definition}[Degree-corrected stochastic blockmodel with covariates]
    Consider a network $\graph = (\node, \edge)$, where each node $i \in \node$ has a covariate vector $x_i$. We call the network follow a degree-corrected stochastic blockmodel with covariates if  \eqref{eqn:dcsbm} and \eqref{eqn:covcluster} are satisfied, with parameter set $(\Theta, K, P, \Pi, F_{[K]}, \calM)$. 
\end{definition}

Under the degree-corrected stochastic blockmodel with covariates, we interpret the definitions of relatively dense communities and sparse communities. Recall $\theta_i$ is the degree heterogeneity parameter of node $i$ and the expected degree $E(d_i) \leq n \theta_i \theta_{\max}$. Let $\theta_{\max} = \|\theta\|_{\infty}$ denote the maximum.
\begin{definition}\label{def:dcsbmcomm}
Consider a network $\graph = (\node, \edge)$ that follows the degree-corrected stochastic blockmodel with parameters $(\Theta, K, P, \Pi)$ and $\theta \in \mathbb{R}^n$ denote the diagonals of $\Theta$ with $\theta_{\max} = \|\theta\|_{\infty}$. 
Consider community $k$, we call it a
\begin{enumerate}
\item {\it (Relatively) dense community}, if there exist constants $c, c_d>0$ so that $\theta_i \geq c \theta_{\max}$ and $n\theta_i\theta_{\max} \geq c_d \log n$ for all $i$ so that $\ell(i) = k$; 
\item {\it Sparse community}, if there exists a constant $0 < c_s < c_d$, so that $n\theta_i \theta_{\max} \leq  c_s\log n$, for all $i$ so that $\ell(i) = k$.
\end{enumerate}
\end{definition}

Let $\calD$ be the set of nodes in relatively dense communities and $\calS$ be the set of nodes in sparse communities. We suppose the node set $\node = \calD \cup \calS$. 
By Definition \ref{def:dcsbmcomm}, all nodes in $\calD$ have expected degree $E(d_i) \asymp n\theta_i\theta_{\max} \geq c_d \log n$. Hence, the diverging $E(d_i)$ indicates sufficient network information. Meanwhile, all nodes in $\calS$ have expected degree $E(d_i) \lesssim n\theta_i \theta_{\max} \leq c_s\log n$. 
In the extremely sparse case, $E(d_i) \to 0$ for $i \in \calS$. It is challenging and the covariates $X$ will be leveraged for accurate clustering.

Most existing works assume that $\theta_i$'s divergence at the same rate and that $n\theta_{\max}^2 \geq c\log n$; see \citet{SC1}, \citet{SCORE} and \citet{krzakala2013spectral}. 
However, some communities in practice tend to make more connections while other communities have more isolated nodes. 
This phenomenon was put forward by \citet{joseph2016impact}, who used  terms ``dense and weak clusters" to refer to communities with different connection intensities, under the stochastic blockmodel. They studied the case where the size of weak clusters does not increase with $n$ and proved the consistency on the dense clusters. By leveraging the covariates, we can achieve exact recovery results even for these sparse communities using our new method.

By the high-level analysis in Section \ref{sec:cascore}, the label recovery of nodes in $\calD$ relies on the network information and of nodes in $\calS$ relies on the covariates. 
Ideally, only the mis-specified nodes in $\calS$ should be affected and mis-clustered, i.e. the set $\calM \cap \calS$. 
Thus we define a parameter as the proportion of these nodes: 
\be\label{misspec}
\epsilon = |\calM \cap \calS|/n. 
\ee
This parameter captures the proportion of the ``essential errors" under the model. Nodes in $\calM \cap \calS$ cannot be labeled correctly, no matter the method applied. We give a theorem to rigorously demonstrate this intuition in Section \ref{sec:lowerbound}. 
One interesting fact is that nodes in $\calM \cap \calD$ are not included, even though they are mis-specified.

\subsection{Consistency of spectral clustering on network-adjusted covariates}\label{sec:consistency}
Under the degree-corrected stochastic blockmodel with covariates, we demonstrate strong consistency of \name. 

To derive the consistency, we first consider the oracle case where model parameters are known. With the population version of $Y$, we derive the left singular matrix $\Xi$ in Lemma \ref{lem:xiformsimple}. It shows the rows of $\Xi$ are closely related to the community labels of corresponding nodes.
We then introduce noise into the model and find $\hat{\Xi}$ in Algorithm \ref{alg1}. In Theorem \ref{thm:dist}, we bound the $\ell_2$ norm between each row of $\Xi$ and $\hat{\Xi}$ up to a rotation.
This row-wise bound leads to the exact recovery of the underlying community memberships on the ``good nodes", i.e. nodes in $(\calM \cap \calS)^c$, as shown in Theorem \ref{thm:rate}.
General forms of the theorems and proofs are available in supplementary materials.

Consider the oracle case where the parameters are known.  
We define the population version of $A, X$ and $D_{\alpha}$. By \eqref{eqn:dcsbm}, $E(A) = \Theta \Pi P \Pi' \Theta$. The oracle matrix for $X$ is that all the nodes display the correct information, which means $\Etilde(x_i) = E(F_{\ell(i)})$ for all $i \in [n]$. Note that  $\Etilde(x_i)$ and $\E(x_i)$ may not be the same if $x_i$ is misspecified. We define $\Etilde(X)$ as the matrix formed by these mean vectors.
Finally, define 
\[
\alpha_i^* = E(\bar{d})\log n/[2\{E(d_i) + \log n\}], \qquad i \in [n],
\]
and $D_{\alpha^*}$ is a diagonal matrix formed by $\alpha_i^*$s.

Now we set up an oracle matrix for the network-adjusted covariate matrix $Y = AX+D_{\alpha}X$. 
Instead of reverting every part to the population version, we consider the dominating terms only. Consider the $i$-th row in $Y$, i.e. $y_i$. If $i \in \calD$, the dominating term is $\sum_{j:A_{ij} = 1} x_j$ and the population version is the $i$th row of $E(A)\Etilde(X)$. Meanwhile, if $i \in \calS$, then the dominating term is $\alpha_i x_i$ and the population version is the $i$th row of $D_{\alpha^*}\Etilde(X)$.  
To express the oracle matrix, let $I_{\calD} \in \mathbb{R}^{n \times n}$ denote the identity matrix where only diagonals on $\calD$ are preserved and others are set as 0. We similarly define $I_{\calS}$. 
The oracle matrix is defined as follows:
\be\label{eqn:oracle}
\Omega = \{I_{\calD}E(A)I_{\calD}\} \Etilde(X) + I_{\calS}\{D_{\alpha^*}\Etilde(X)\} = \{I_{\calD} E(A) I_{\calD} + I_{\calS} D_{\alpha^*}\}\Etilde(X). 
\ee

The nodes can be decomposed into 3 disjoint sets: $\calD$ where nodes are in relatively dense communities; $\calS \cap \calM^c$ where nodes are in sparse communities with correct covariate distribution; and $\calS \cap \calM$ where nodes are in sparse communities and the covariates are mis-specified. Only nodes in the first two sets are possible to be correctly clustered. Theorem \ref{thm:lowerboundnew} in below rigorously states the impossibility for nodes in $\calS \cap \calM$. 
Hence, we define the set of ``good" nodes as follows: 
\be\label{eqn:good}
\cG = \calD \cup (\calS \cap \calM^c). 
\ee
Compare it with \eqref{misspec} and we can see $\epsilon = 1 - |\cG|/n$.
We then define $I_{\cG}$ in a similar way as $I_{\calD}$. 

\begin{lemma}\label{lem:xiformsimple}[Spectral analysis of the oracle matrix]
Consider the oracle matrix $\Omega$.
with a good node set $\cG$. 
Let $\OX_{\cG} = I_{\cG}\Omega$ be the oracle matrix with rows restricted on $\cG$.
Denote the singular value decomposition of $\OX_{\cG}$ as
$\OX_{\cG} = \Xi \Lambda U'$,
where $\Xi \in \mathbb{R}^{n \times K}$, $U \in \mathbb{R}^{p \times K}$ and $\Lambda \in \mathbb{R}^{K \times K}$. 

Under the degree-corrected stochastic blockmodel with covariates, there is 
\[
\Xi_i=\left\{\begin{array}{ll}
\theta_i v_{\ell(i)},& i\in \calD,\\
\alpha^*_i u_{\ell(i)},& i\in \calS\cap \calM^c,\\
0,& i\in \calS \cap \calM = \cG^c,\\
\end{array}
\right.
\]
where $v_k$'s and $u_k$'s  are $K$-dimension vectors. 
\end{lemma}

Lemma \ref{lem:xiformsimple} shows that nodes in the same community share the same rows in $\Xi$, up to a constant factor. This constant factor is the degree heterogeneity parameter $\theta_i$ for $i \in \calD$ or the weightage $\alpha^* \approx E(\bar{d})/2$ for $i \in \calS \cap \calM^c$. The explicit formula for $u_k$ and $v_k$ can be found in supplementary materials. Normalizing $\Xi_i$ to be of unit length removes the constant factor.
When the centers (normlized $u_k$s and $v_k$s) are well separated, the labels of nodes in $\cG$ can be exactly recovered.
\begin{theorem}\label{thm:dist}[Row-wise empirical and oracle singular matrix distance]
Consider the degree-corrected stochastic blockmodel with covariates with parameters $(\Theta, K, P, \Pi, F_{[K]}, \calM)$, where $p >0$ is a constant, $\cG$ is the set of good nodes and $\epsilon = |\cG^c|/n$. Let $\Omega$ be the oracle matrix defined in \eqref{eqn:oracle}, $\Xi$ is the left singular matrix of $\Omega_{\cG}$ and $\hat{\Xi}$ consists of the top $K$ left singular vectors of $Y$. 

Let $c,C > 0$ be constants that vary case by case. 
We assume (i) the sub-matrix of $P$ that restricted to dense communities $P_{\calD}$ is full-rank; 
(ii) $\|x_i\|\leq \munorm$ almost surely, and for $i \in \calS \cap \cG$, with high probability $\|x_i- \Etilde(x_i)\|\leq \delta_X \munorm$;  (iii) $\lambda_{K}\{\Etilde(X)\} \geq c\sqrt{n}\munorm$; and (iv) the number of nodes in any community $n_k/n \geq c > 0$. Then there are threshold constants   $C_\theta$, $\epsilon_0, n_0$, and $\delta_0$,  so that if $\delta_X \leq \delta_0, \epsilon\leq \epsilon_0, n\geq n_0, n\theta_{\max}^2 \geq C_\theta \log n$,  there exists an orthogonal matrix $O$ and a constant $C > 0$ with probability $\q$ so that 
\[
\max_{i\in \cG}\|\Xihat_{i} - O\Xi_i \|\leq C(\delta_X+\sqrt{\epsilon} + 1/\sqrt{C_{\theta}})/\sqrt{n},
\]
where $\Xihat_{i}$ and $\Xi_i$ are vectors formed by $i$th row of $\Xihat$ and $\Xi$. 
\end{theorem}
Assumptions (i) and (iv) are regular conditions on networks and Assumptions (ii) and (iii) impose regularity conditions on the covariates. 
Assumption (i) requires linearly independent rows/columns of $P_{\calD}$ for different dense communities, a common requirement in community detection on $P$ \citep{SCORE, FengNodal}. We use the network information only for $\calD$ so the requirement is on $P_{\calD}$. 
Assumption (iv) is a standard requirement in stochastic blockmodels to ensure comparable community sizes. 
Assumption (ii) states the range of covariates $x_i$ and the concentration. Bounded covariates are preferred (same as \cite{casc}), although the results can be extended to unbounded cases, with more complicated interpretations. 
Nodes in sparse communities are labeled by their covariates, so for those nodes, the noise or deviation $\|x_i - \Etilde(x_i)\|$ must be small. The noise level is denoted by $\delta_X$. 
Assumption (iii) requires the smallest singular value of $\Etilde(X)$ to be $O(\sqrt{n}R)$, the same order as the largest singular value. It can be verified if the $K \times p$ matrix formed by $E(F_k)$'s is non-singular. On the other hand, if $\Etilde(X)$ has a rank $< K$, then even the oracle matrix $\Omega$ has a rank $< K$ and $\Xi$ can be insufficient for label recovery.
Finally, the results hold when the relatively dense communities have degrees exceeding $C_{\theta}\log n$; otherwise the network does not contribute any information. 

Theorem \ref{thm:dist} provides a row-wise bound for the closeness between the rows of $\Xihat$ and $\Xi$, up to a rotation. 
It supports the row-wise operations (normalisation and $k$-means) in Algorithm \ref{alg1}. The most notable aspect of the theorem is that the bound is row-wise, instead of the standard Frobenius norm bound on $\Xihat - \Xi O$ using the Davis-Kahan approach. To our best knowledge, this result is not proved in the context of degree-corrected stochastic blockmodel with covariates, and it cannot be obtained using existing analysis tools under our challenging mis-specified setting. A new singular vector ``stacking" result is developed in the supplementary material to handle mis-specification.
This row-wise analysis enables the separation of good nodes $\cG$ and other nodes $\cG^c$, leading to exact recovery on $\cG$, even in the presence of mis-specification. 

The bound  given in Theorem \ref{thm:dist} contains three parts: $1/\sqrt{C_{\theta}}$ due to density of the network, $\delta_X$ due to the randomness in $X$, and $\sqrt{\epsilon}$ due to the mis-specified nodes in sparse communities, i.e. $\cG^c$. It is worth noting that although only nodes in $\cG$ are considered, perturbation from $\cG^c$ is unavoidable since the singular vectors are obtained from the whole matrix. When all three terms are well bounded, we achieve exact recovery of the entire network.

\begin{theorem}\label{thm:rate}
[Strong consistency of Algorithm \ref{alg1}]
Suppose conditions (i)--(iv) in Theorem \ref{thm:dist} hold. 
Let $\hat{\ell}$ be the estimated labels by the spectral clustering method on $Y$. 
Then there is a constant $C_\theta$ independent of $n$, so that if $n\theta_{\max}^2 \geq C_\theta \log n$,  
there exists a permutation $\pi$ so that with probability $\q$, 
\[
\pi\bigl\{\hat{\ell}(i)\bigr\}=\ell(i),\quad \mbox{all } i\in \cG. 
\]
Therefore, the community detection error rate is bounded by $|\cG^c|/n = \epsilon$.
\end{theorem}
In summary, under reasonable conditions, with a high probability, our method can exactly recover the label of every node in $\cG$.

\subsection{Consistency of spectral clustering on generalized network-adjusted covariates}\label{sec:ainclude}
To achieve strong consistency, we require $\Etilde(X)$ to have rank $K$. In Algorithm \ref{alg2}, we consider the case that $rank(\Etilde(X)) < K$ and add $AA'$ to $YY'$ to achieve good clustering results. In this section, we demonstrate the consistency of this algorithm.

Consider the multiscale network with covariates. Nodes in $\calD$ are labeled by the network information, for which we do not expect conditions on $X$. Nodes in $\calS$ are labeled by the covariate information, and distinction is required on these $x_i$'s. Hence, we relax condition (iii) to be (iii') on only the sparse communities, as follows: 
\begin{enumerate}
\item[(iii')] Let $K_{\calS}$ be the number of sparse communities. There is a constant $c > 0$, so that  $\lambda_{K_{\calS}}\{I_{\calS}\Etilde(X)\} \geq c\sqrt{n}R$. 
\end{enumerate}

To demonstrate consistency of Alg. \ref{alg2}, we first introduce $\OXtilde$ as the counterpart of $\Omega$ in the generalized case. The contribution from $AA'$ is summarized in a diagonal matrix $T = (\Pi'\Theta \Pi)^{-1}(n\Pi'\Theta^2 \Pi)^{1/2}\in \RR^{K_{\calD} \times K_{\calD}}$, where each entry is the $\ell_2$ norm divided by $\ell_1$ of the degree heterogeneity vector restricted to one community. Let $T_k$ denote the row of $T$ corresponding to community $k$. 
Define new ``extended covariates" as follows
\[
\tilde{x}_i = \left\{\begin{array}{ll}
 (x_i, \sqrt{\beta}T_{\ell(i)}),    &  i \in \calD,\\
 (x_i, 0),    &  i \notin \calD.
\end{array}
\right.
\]
Denote the matrix of extended covariates as $\tilde{X}$, and correspondingly $\Etilde(\tilde{X})$. 

Therefore, we have $\OXtilde$ as 
\be\label{eqn:oraclegen}
\OXtilde = \bigl\{I_{\calD} E(A) I_{\calD} + I_{\calS} D_{\alpha^*}\bigr\}\Etilde(\tilde{X}).
\ee
The definition of $\OXtilde$ is almost the same with $\OX$, except replacing $\Etilde(X)$ to the extended covariates $\Etilde(\tilde{X})$. 
For $\OXtilde$, there is $\OXtilde\OXtilde'=\OX\OX'+\beta n\{I_{\calD} E(A) I_{\calD}\}^2$, which concludes the oracle matrix corresponding to $YY'$ and the dense communities in $\beta n AA'$.

Following the same proof, we have a generalized version of Lemma \ref{lem:xiformsimple} for $\OXtilde$, where the spectral matrix of $\OXtilde\OXtilde'$ have identical rows for nodes in the same community, up to a constant factor. We further investigate the difference between $\Xihat$ and $\Xi$ in terms of the Frobenius norm. Based on it, we prove the consistency of community detection by using $\Xihat$.

\begin{theorem}\label{thm:dist2}[Consistency of Algorithm \ref{alg2}]
Consider the degree-corrected stochastic blockmodel with covariates with parameters $(\Theta, K, P, \Pi, F_{[K]}, \calM)$.
Let $\Xihat$ be the matrix formed by top $K$ eigenvectors of $YY' + \beta n AA'$ in Algorithm \ref{alg2} and $\Xi$ be the eigenvectors of $\OXtilde \OXtilde'$. 
Suppose conditions (i), (ii), (iv) in Theorem \ref{thm:dist} and (iii') hold. 
Then there are threshold constants   $C_\theta$, $\epsilon_0, n_0$, and $\delta_0$,  so that if $\delta_X \leq \delta_0, \epsilon\leq \epsilon_0, n\geq n_0, n\theta_{\max}^2 \geq C_\theta \log n$, \blue{there are constants $\beta_0$ and $\beta_1$, so that when $\beta_0 < \beta < \beta_1$,} with probability $1 - O(1/n)$, 
there exists an orthogonal matrix $O$ and a constant $C > 0$ that 
\[
\|\Xihat - O\Xi \|_F\leq C(\delta_X+\sqrt{\epsilon} + 1/\sqrt{C_{\theta}}).
\]
Let $\delta_{net} = \frac{\max_{i \in \calS} n\theta_i\theta_{\max}}{\min_{i \in \calD} n\theta_i\theta_{\max}}$. 
Let $Err_n = n^{-1}\min_{\pi: [K]\to[K]}|\{i:\pi(\lhat(i))\neq \ell(i)\}|$. 
Then there exists a permutation $\pi$, so that with probability $\q$, the clustering error rate by Alg. \ref{alg2} follows 
\[
Err_n \leq C(\delta_X +1/\sqrt{C_\theta} + \delta_{net} + \sqrt{\epsilon}). 
\]
\end{theorem}
\begin{remark}
Theorem \ref{thm:dist2} gives the weak consistency results of Algorithm \ref{alg2}.    While both algorithms share a similar format of the oracle matrix, the noise caused by $AA'$ is relatively large. 
It comes from the nature of Bernoulli distribution. The large noise causes the row-wise control on the empirical singular matrix to be hardly available, and hence strong consistency not achievable. 
It further provides the importance of introducing covariates into community detection. 
\end{remark}
\begin{remark}
    The error control contains three parts: $\delta_X + 1/\sqrt{C_{\theta}}$ from the covariates, $\delta_{net}$ from the network, and $\sqrt{\epsilon}$ from the mis-specification. 
    For a consistent signal recovery, we require the covariates have relatively small deviations, the network is sufficiently dense and the sparse communities are well-separated from the dense communities, and there are no mis-specifications. 
\end{remark}
\begin{remark}
    \blue{The tuning parameter $\beta$ is a fixed constant. The requirement on $\beta$ depends on $K_{\calS}$, i.e. the number of sparse communities. 
    When $K_{\calS} = 0$, the information is in $AA'$, so larger $\beta$ is always preferred. Delicate analysis shows that $\beta \gtrsim cR^2$ is required, where $c$ is a small constant based on $\delta_X, \delta_{net}, C_{\theta}, \epsilon$ and $P$. When $K_{\calS} > 0$, $\beta$ must be in an interval $(\beta_0, \beta_1)$ so that $\beta n AA'$ is sufficiently large to detect dense communities, and not too large so that $Y$ still works for the sparse communities. 
    In supplementary materials, we find that the detailed requirement on $\beta_0$ and $\beta_1$, where both can be represented as $c\lambda_{K_{\calS}}^2[I_{\calS}\{\Etilde(X)\}]/n \approx cR^2$ by (iii'). For both cases, $\beta \approx cR^2$ works.}
\end{remark}

\subsection{Statistical lower bound on networks with covariates}\label{sec:lowerbound}
In this section, we demonstrate the information bound of the community detection problem on multiscale networks with covariates. We focus on the exact recovery problem. Our goal is to find the region of $\theta_i$ and $F_k$, in which all estimators will fail.

To capture the effects of $\theta_i$ and $F_k$ for a multiscale network, we need at least 3 communities: one of them is relatively dense and the other two are extremely sparse. If there is only one sparse community and one dense community, then the labels can be recovered by the degree distribution, which is not of interest.  

Consider a simplified model $SM(\theta_0, \theta_{\max}, P, \mu_{[K]}, \sigma)$ with $K = 3$ communities. Under this model, nodes fall into each community equally likely. Further, there are only two possible values of $\theta_i \in \{\theta_0, \theta_{\max}\}$. 
Nodes in communities 1 and 2 have $\theta_i = \theta_0$ and nodes in community 3 have $\theta_i = \theta_{\max}$. Therefore, community 3 will be a dense community and communities 1 and 2 have the flexibility to be either dense or sparse. 
The covariates follow $x_i \sim \mathcal{N}(\mu_{\ell(i)}, \sigma^2 I_p)$. Here, $\sigma^2$ is to capture the deviation from the mean vector. 
We find the upper bound decided by our new algorithm and the statistical lower bound under the model $SM(\theta_0, \theta_{\max}, P, \mu_{[3]}, \sigma_n)$. 

\begin{theorem}[Statistical Lower Bound]\label{thm:lowerboundnew} 
Consider the simplified model with $K = 3$ communities, denoted $SM(\theta_0, \theta_{\max}, P, \mu_{[3]}, \sigma_n)$. 

There is a constant $C > 0$ and a constant $c_p$ on $P$, so that 
if (i) $n\theta_0\theta_{\max} < Cc_p\log n$ and (ii) $\|\mu_2 - \mu_1\|/\sigma_n < \sqrt{C\log n}$, 
then for any estimator $\hat{\ell}$, 
\[
E\{Err_n(\hat{\ell}, \ell)\} \geq 1/n, 
\]
i.e. the exact recovery cannot be achieved.
\end{theorem}
By Theorem \ref{thm:lowerboundnew}, to achieve the exact recovery, the communities either have a diverging degree as (i), or have well-separated covariates as (ii). When neither is met, all estimators will fail. 
\begin{remark}[Connection to the upper bound]
Compare the lower bounds in Theorem \ref{thm:lowerboundnew} with the upper bounds in Theorem \ref{thm:rate}. The exact recovery can be guaranteed when all communities are dense or all nodes in the sparse communities have covariates $\|x_i - \Etilde(x_i)\|\leq \delta_X R$ for a small constant $\delta_X$. Under $SM(\theta_0, \theta_{\max}, P, \mu_{[3]}, \sigma_n)$, it means either $n\theta_0 \theta_{\max} \geq C\log n$ or $\|\mu_2 - \mu_1\|/\sigma_n \geq C \sqrt{\log n}$.
Therefore, the upper bound in Theorem \ref{thm:rate} meets the lower bound in Theorem \ref{thm:lowerboundnew}, up to a constant factor.  It supports the optimality of our  \names approach.    
\end{remark}

\section{Simulation}\label{sec:simulation}
Consider a large network with $n = 1200$ nodes and the number of covariates can be either small ($p = 20$) or large ($p = 600$). We conduct 3 sets of simulation studies that focus on the effects on the estimation when: (1) the network changes; (2) the signal strength changes; and (3) the proportion of mis-specified nodes. 

We set up the baseline degree-corrected stochastic blockmodel with covariates as follows. 
The network has $K = 4$ communities and $\ell(i) = k$ with equal probability for $k = 1,2,3,4$. The first two are dense communities with $\theta_i \sim Unif(0.3, 0.5)$ and the last two are sparse communities with $\theta_i \sim Unif(0.03, 0.05)$. The intensity matrix $P$ is 
\be
P = \left(\begin{array}{llll}
1  & \alpha \quad & \alpha & \alpha\\
\alpha \quad & 1 & \alpha & \alpha\\
\alpha \quad & \alpha \quad & 1 & \alpha\\
 \alpha\quad & \alpha \quad & \alpha \quad & 1
\end{array}
\right).
\ee

The covariates follow a mixture of five distributions $F_k$, $k = 1,2,3,4,5$. Here $F_5$ only appears for some of the mis-specified nodes.
In detail, the covariates of node $i$ follows the mixture distribution $x_i \sim (1-\gamma) F_{\ell(i)} + \sum_{k= 1, k\neq \ell(i)}^5 (\gamma/4) F_{k}$. 
Hence, $\gamma$ represents the fraction of mis-specified nodes. Half of these mis-specified nodes belong to dense communities 1 and 2, so the ``bad nodes" proportion is about $\epsilon = |\cG^c|/n = \gamma/2$.

We set $F_k \sim \mathcal{N}(m_k, I_p)$ for the covariates and discuss the setting of mean vectors $m_k$. 
In the large covariates case, $p = 600$, most covariates are noise with mean $0$. 
We first select $5\%$ of the $p$ covariates as the ``useful" covariates and then set the mean for those covariates as $m_k(j) \sim \mu_1*Bernoulli(1/2)$ independently. Therefore, $m_k(j)$ is different across communities $k$, which carries the label information. The leftover covariates have mean $m_k(j) = 0$. 
In the small covariates case, $p = 20$.
Let $m_k(j) \sim \mu_2 + 0.1*Bernoulli(0.5)$ when $j \in \{5k-1,5k-2,5k-3,5k-4\}$ and $m_k(j) \sim 0.1*Bernoulli(0.5)$ otherwise.


\begin{figure}
    \includegraphics[width = 1.0\textwidth]{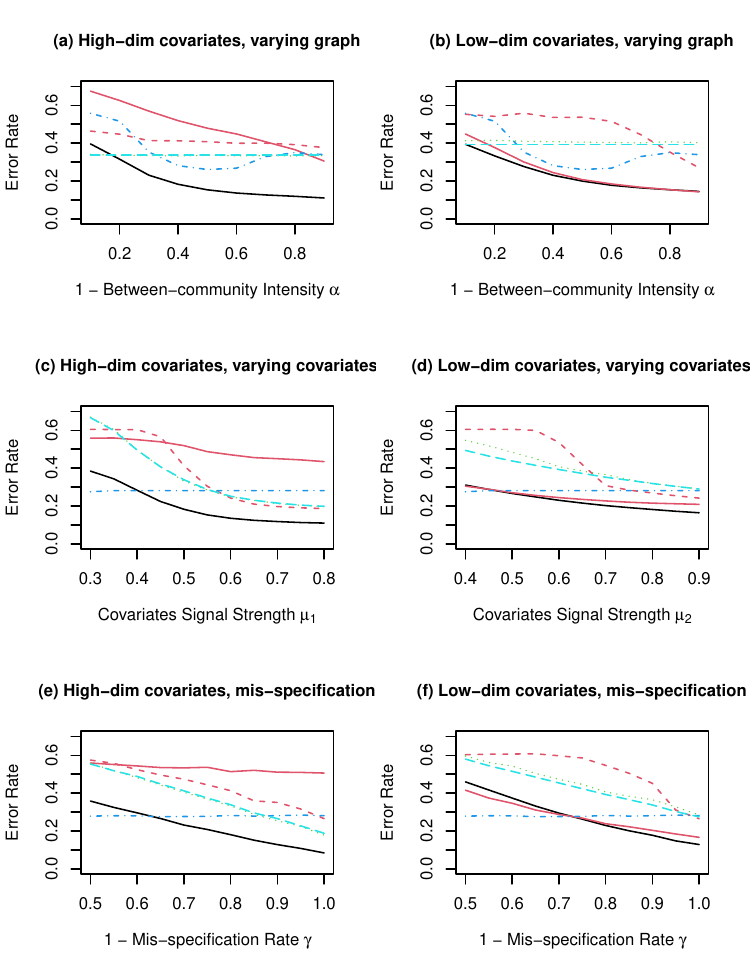}
\caption{The average clustering error rate among 50 repetitions of 6 community detection methods: (1) our method (black solid); (2) our method on generalized covariates (red solid); (3) covariate-assisted Laplacian (dashed); (4) semidefinite programming (dotted); (5) spectral clustering on regularized Laplacian (dotdash); and (6) spectral clustering on the covariate matrix (longdash). 
The fixed parameters are $n = 1200$,  $\alpha = 0.4$, $\mu_1 = 0.5$, $\mu_2 = 0.8$, and $\gamma = 0.2$. }\label{fig:exp}
\end{figure}

For each simulation study, we compare six different methods. The first four are popular community detection methods based on both network and covariates: our newly proposed 
spectral clustering on network-adjusted covariates method by Algorithm \ref{alg1}, our new algorithm on generalized covariates by Algorithm \ref{alg2},
the covariate-assisted spectral clustering in \cite{casc}, and the semi-definite programming in \cite{attrisparse}. We also consider spectral method on the network only, in particular the spectral clustering with regularized Laplacian method in \cite{joseph2016impact}.
To examine the effects using only the covariate matrix, the last method is the spectral clustering method on $XX'$ in \cite{specgem}.

The first simulation study is on the community by community connection intensity matrix $P$. Let the between-community intensity $\alpha \in [0.1, 0.9]$. 
A smaller $\alpha$ will cause larger differences between the within-community connections and between-community connections, so that the community detection is easier. On the other hand, a very small $\alpha$ indicates many low-degree nodes in the sparse communities. Hence, there is a sweet spot of $\alpha$ when the community detection is purely based on the network, as shown in Fig. \ref{fig:exp}(a) and (b). 
Our approach outperforms all other methods in most cases, except the semi-definite programming method and covariate-based method when the network information is very few.

The second simulation study is on the signal strength in covariates. We set the signal strength $\mu_1 \in [0.3, 0.8]$ for $p = 600$ and $\mu_2 \in [0.5, 1.0]$ for $p = 20$. Fig. \ref{fig:exp}(c) and (d) record the results for $p = 600$ and $p = 20$, respectively. 
In general, increasing the signal strength can improve the community detection results for all methods including covariates. Among all methods, the new approach always performs the best, while the error rates of all other methods on covariates are at least $0.19$ that may come from the $20\%$ mis-specified nodes.  

In the third simulation study, we focus on the mis-specification rate $\gamma$. Let $\gamma \in [0, 0.5]$ and the results are in Fig. \ref{fig:exp}(e) and (f).
When $\gamma$ is large and the network adjacency conveys strong information (the low-dimensional case), the network-based method and covariate-assisted Laplacian work better. The reason is large mis-specification rate will cause centers of community covariates close to each other. 
When $\gamma$ decreases, our method works best. Especially, the error rate of our method decreases at the rate of $\gamma/2$, which is roughly $\epsilon = |\cG^c|/n$.
It provides numerical support that our method fails on the mis-specified nodes with low degrees.

\section{Real World Networks}\label{sec:data}

 \subsection{Error rates on the LastFM network}
 LastFM Asian dataset is a social network of LastFM app users. This dataset was collected and cleaned by \citet{LastFM} from the public API in March 2020. 
 The nodes are the app users from 18 unspecified Asian countries and the connections between them are identified by the mutual friendship. In pre-processing, a small country with only 17 users is removed because of insufficient information and 17 countries are left. Each user has a list of liked artists as covariates. The goal is to estimate the country membership for each node. 

We consider 4 sub-datasets from the original dataset: small country dataset (each with less than 100 users), medium country dataset (each with users between 100 and 300), large-sized country dataset (each with users between 300 and 1000) and huge country dataset (each with more than $1000$ users). 
Each dataset has communities with similar sizes, i.e. non-degenerate communities that is required by most existing community detection methods. 
For each dataset, we first select the ``regional popular" artists by examining the artists with largest proportion of fans in the countries of interest, and then apply five community detection methods.

\begin{table}
\def~{\hphantom{0}}
{%
\begin{tabular}{lcccccccccc}
 \\
 & $n$ & $p$ & $K$ & $\bar{d}$ & New & New (Generalized) & CAL &SDP & Net-based & Cov-based \\[5pt]
{Small countries} & 343 & 194 & 6 & 4.9 &  .236 & .262 & { .178} & .510  & .350 & .557   \\
Medium countries & 612 & 324 & 3  & 7.2 & .041    &  {.031}  & .044 & .342 & .165 & .234 \\
Large countries & 2488 & 600 & 5  & 6.2 & {.249} & .424 & .371    & .519 & .482 &.468  \\
Huge countries & 3691 & 600 & 3  & 7.1 & {.019} & .022 & .019    & .416 & 
.240 & .328
\end{tabular}}
\label{tab:LastFM}
\caption{
$n$, number of nodes; $p$, number of covariates; $K$, number of communities; $\bar{d}$, the average degree. New, our new approach; New (Generalized), our new approach on generalized covariates; CAL, covariate-assisted Laplacian method; SDP, semidefinite programming method; Net-based, spectral clustering on regularized Laplacian matrix; Cov-based, spectral clustering on the covariate matrix.}
\end{table}

The sizes of networks and covariates are summarized in Table \ref{tab:LastFM}, together with the average degree and community detection error rates of the methods. 
Our new method outperforms all other methods on 3 out of 4 datasets.
For the covariate-assisted Laplacian method, we selected the optimal tuning parameter among 5 choices according to \cite{casc}.  
Our network-adjusted covariate based community detection method does not need any tuning parameter, with comparative clustering error rates. 

In Table \ref{tab:LastFM}, the average degree is stable when the network size changes from 343 to 3691. This is often seen in real networks, and it motivates us to investigate sparse networks. 
When the network sizes are small and the average degree is relatively large, the Net-based method performs well except for the low-degree nodes. Combining it with the covariates further reduces such errors; see our new methods and covariate-assisted Laplacian method. Here, the covariate-assisted Laplacian method outperforms ours for the small countries data set, because this dataset has severe degree heterogeneity within the same community and uninformative covariates.  
When it comes to large/huge countries, the network is relatively sparse. Our new method performs the best on these two datasets.

\subsection{Community detection on statisticians' citation network}\label{sec:citation}
The citation network was published in \citet{network}. It contains 3232 papers published in the Annals of Statistics, Journal of American Statistical Association, Journal of Royal Statistical Society (Series B) and Biometrika from 2003 to the first half of 2012. Each paper is a node and two papers are connected if they both cite the same paper or are both cited by another paper. The covariate of each paper is its abstract. 
This citation network has $n = 3232$ nodes and $p = 4095$ covariates.

In this network, there is a giant component of 2179 nodes. Among the leftover 1053 nodes, 957 are isolated and 96 nodes are in small components with the largest size of 9. 
When it was introduced in \citet{network}, the authors applied spectral clustering on ratios-of-eigenvectors in \citet{SCORE} to the network. This method requires the network to be connected, hence only the giant component was analyzed. It suggests interpretable results on the high-degree nodes, yet the 1053 nodes outside the giant component were not classified. 

We apply our new method, the variant of our new method on the generalized covariates, covariate-assisted Laplacian in \cite{casc}, semidefinite programming method in \cite{attrisparse}, and the network-based method in \cite{joseph2016impact}. For all the methods, we take $K = 5$ and record the community sizes in Table \ref{tab:citation1}. 
For the low-degree nodes, the network-based method turns to classify all of them into the largest community, which is unrealistic. This issue still exists but is slightly resolved for our method on the generalized covariates and the covariate-assisted Laplacian method, because of their dependency on the network. Our original method (Alg. \ref{alg1}) and the semidefinite programming method have a more reasonable and balanced splitting.

\begin{table}[!htbp]
\def~{\hphantom{0}}
{%
\begin{tabular}{llllll}
 &  Community 1 & Community  2 & Community  3 & Community  4 & Community  5\\
New Method  & 1105(345)  & 1062(386)  & 483(181)  & 325(120) & 257(21)  \\
New (Alg. \ref{alg2}) & 2433(1032)  &  228(8) &  225(1)  & 180(12)  & 166(0)  \\
SDP & 783(271)  & 720(226)  & 665(200)   &  660(217)  & 404(139) \\
CAL & 1892(1047) & 471(4)  & 357(0)   & 319(1)  & 193(1)  \\
Net-Based & 2280(1053)  & 297(0)  & 283(0)  & 221(0) & 151(0)  \\
\end{tabular}}
\label{tab:citation1}
\caption{
Estimated community sizes and the number of the leftover 1053 nodes in each community (in the bracket).}
\end{table}

We then investigate the splitting of the giant component. We measure the agreement of each pair of clustering results by the normalized mutual information (NMI) in \citet{nmi}. A larger NMI score means the two clustering results are more coherent. 
The heatmap of pairwise NMI scores is in Fig. \ref{fig:citation}. Considering the giant component only, our new method agrees with the network-based method, while the semidefnite programming method doesn't agree with the network-based method at all. 
Combining the giant component and the low-degree nodes, our method provides the best splitting.

We compare the communities found by our spectral clustering on network-adjusted covariates method with the communities found in \cite{network}. 
For each community, we check the top 10 popular papers and corpus. The results can be interpreted as: Variable selection (regression) community; Variable selection (semi-parametric) community; Large-scale multiple testing community; Biostatistics community; and Bayesian community. Compared to the estimated communities of statisticians in \citet{network}, three communities are coherent: large-scale multiple testing community, biostatistics community, and Bayesian community (non-parametric community in \citet{network}). The variable selection community in their work has been decomposed into two communities by our new method, one about regression and one about the semi-parametric models. In Fig. \ref{fig:citation}, we can see the regression community and the semi-parametric community are densely connected but have an even denser connection within communities.

\begin{figure}
    \includegraphics[width = 1.0\textwidth]{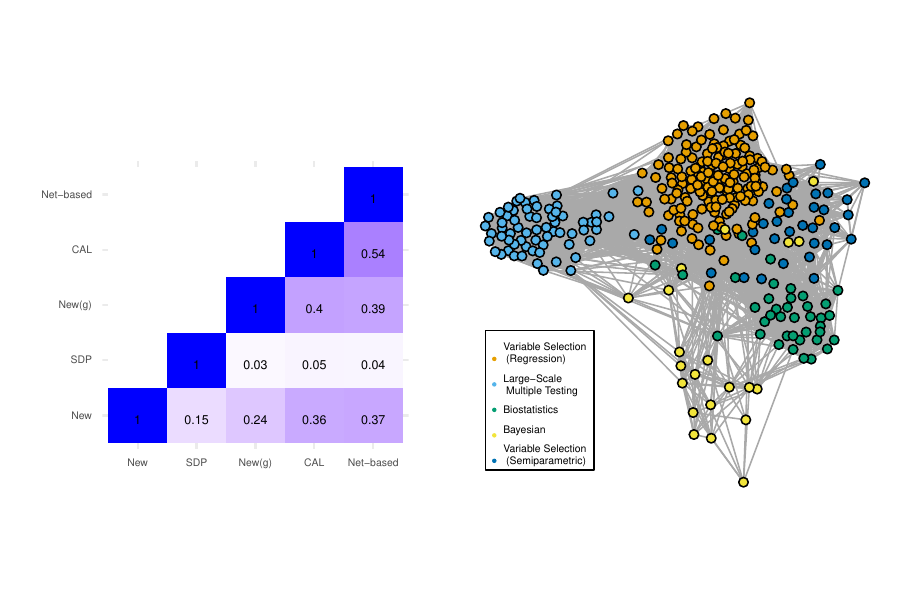}
\caption{Left: The heatmap of NMI scores between 5 community detection methods on the giant component. Right: A sub-network consisting of nodes with at least 50 neighbors. Each color denotes one estimated community.}\label{fig:citation}
\end{figure}


The interpretation of communities  works on the isolated nodes. 
We randomly checked two papers without any connections. Node 2893 titled ``Bayesian pseudo-empirical-likelihood intervals for complex surveys \citep{rao2010bayesian}" is classified into the Bayesian community. Node 2481 titled ``Testing dependence among serially correlated multicategory variables \citep{pesaran2009testing}" is classified into the Multiple Testing community. 
These two examples suggest our method reasonably clustered isolated nodes.

\section*{Acknowledgement}
This research was supported by Singapore Ministry of Education Academic Research Fund Tier 1 A-00004813-00-00.
The authors thank Purnamrita Sarkar and Bowei Yan for generously sharing their code and thank Xin T. Tong, Yang Feng, Ramon van Handel, Jiashun Jin, Liza Levina, Anru Zhang and the anonymous referees for their valuable comments and discussions.
\section*{Supplementary material}
\label{SM}

Supplementary material available at \Bka\ online includes additional results on the statistician citation network, theoretical proofs of theorems and additional theoretical results. R code is available at \url{https://github.com/YaofangHuYaofang/NAC}.

\vspace*{-10pt}

\bibliographystyle{biometrika}
\bibliography{ref}

\end{document}